\documentclass[journal,12pt,onecolumn,draftclsnofoot,]{IEEEtran}

\usepackage{amsmath}
\usepackage{amsfonts}
\usepackage{amssymb}
\usepackage{graphicx}
\usepackage{tikz}
\usepackage{amsthm}
\newtheorem{theorem}{Theorem}

\newtheorem{corollary}{corollary}

\usetikzlibrary{shapes.misc}

\usepackage[tight,footnotesize]{subfigure}

\tikzset{cross/.style={cross out, draw=black, minimum size=2*(#1-\pgflinewidth), inner sep=0pt, outer sep=0pt},
	cross/.default={6pt}}

\usetikzlibrary{matrix,chains,positioning,decorations.pathreplacing,arrows}
\usetikzlibrary{positioning,calc}

\ifCLASSINFOpdf
\else
\fi

\begin{document}
\title{Secure Energy Efficiency: Power Allocation and Outage Analysis for SWIPT-in-DAS based IoT}

\author{\IEEEauthorblockN{Aaqib Bulla,} and \IEEEauthorblockN{Shahid M Shah}
\IEEEauthorblockA{\\Communication Control \& Learning Lab\\
Department of Electronics \& Communication Engineering\\
National Institute of Technology, Srinagar\\
Email: aaqib.bulla@gmail.com,shahidshah@nitsri.ac.in}}

\maketitle

\begin{abstract}
%\boldmath
In this paper we study secure energy efficiency (SEE) for simultaneous wireless information and power transfer (SWIPT) in a distributed antenna system (DAS) based IoT network.  We consider a system in which both legitimate users (Bobs) and eavesdroppers (Eves) have power splitting (PS) receivers to simultaneously decode information and harvest energy from the received signal.  When the channel state information (CSI) is known at the transmitter, we analyze the effect of an energy harvesting eavesdropper (EHE) over the maximization of SEE of the system. Next, considering the fact that perfect CSI is hard to achieve in practice, we characterize the system performance in terms of the outage probability of SEE. For the given SWIPT-in-DAS setup, we derive the closed form expression for the outage probability of SEE and with the help of numerical results, we study the effect of transmit power levels, number of distributed antenna (DA) ports and the PS ratio of devices. To the best of our knowledge, this is the first attempt to define the outage probability of SEE for SWIPT-in-DAS.
\end{abstract}


% \fi
\IEEEpeerreviewmaketitle

\section{Introduction}
\label{sec1}
While the next generation (5G/6G) communication systems promise to meet the ever rising demands for high data rates, security and quality of service, there is a challenge of tremendous energy consumption due the exponential growth in the number of IoT devices \cite{gandotra2017green}.\\
Distributed antenna system (DAS) technology, primarily designed for
increasing network coverage and data rates, is now being
studied in the field of energy efficient wireless communication \cite{kim2014optimal,8863972}. Since DAS reduces the transmitter-receiver access distance, it can significantly help in SWIPT, which is expected to be an energy efficient alternative to facilitate the battery-less operation of IoT devices \cite{huang2018energy}. Moreover, physical layer (PHY) security is also being widely studied alongside energy efficient wireless communication in IoT \cite{9071996,8382162}.\\
Motivated by the aforementioned technologies and observations, in this paper, we study energy efficiency for SWIPT in a DAS based IoT network, wherein, the information is kept confidential at the physical layer. Although, SEE for SWIPT-in-DAS with perfect CSI has been investigated in literature \cite{8382162}, we aim to analyze the effect of eavesdropper charge constraint over the maximization of SEE in the same scenario with known PS ratios. We use the secrecy rate metric in \cite{wyner1975wire} to define secure energy efficiency (SEE) as ratio of the secrecy
rate to the total power consumed at DA-ports.  We formulate the maximization of SEE as a constrained fractional optimization problem and obtain the optimal solution by solving KKT conditions.\\  
In next part of the paper, we consider a more practical scenario, wherein, the CSI (of both Bob and Eve) is not available at the transmitter.   To characterize the system performance, we define outage probability (OP) of SEE. Considering the blanket transmission scheme, wherein all the DA ports are active, we derive the closed form expression for the OP of SEE. Further, we also inquest a more general case, wherein, multiple Eves are present in the system and evaluate the OP of SEE corresponding to the worst case secrecy rate achievable for the given IoT device.\\
The rest of this paper is organized as follows: In section 2, we discuss the system model and problem formulation of the optimal power allocation with energy harvesting eavesdropper. In section 3 we study the outage probability of SEE. The numerical results are discussed in section 4. Finally, section 5 concludes the paper.

\section{Optimal power allocation with energy harvesting Eve}\label{sec2}
\subsection{System model}
%\begin{figure}[h]
%	\centering
%	\includegraphics[width=.7\columnwidth]{das.png}
%	\caption{A DAS model with $N=3$}
%\end{figure}
Let’s consider a downlink DAS with $N$ centrally controlled DA ports serving $K_b$ number of users in presence of an EHE, all equipped with single antenna. The signal received by a given device in such a setup is written as; $y=\sum_{i=1}^{N}\sqrt{p_i}\zeta_ix_i+n$ \cite{kim2014optimal}, where, for $i^{th}$ DA-port, $p_i$ is the transmit power, $x_i$ denotes the transmitted symbol with average power $E{[|x_i|^2]}=1$ and $\zeta_i$ is the corresponding fading co-efficient. Also, $n$ denotes the additive white Gaussian noise (AWGN) at the receiver. In the SWIPT system, each IoT device has a power splitter which splits the received signal power according to a power splitting (PS) ratio  ($\Delta$ for information decoding and the rest $1-\Delta$ for energy harvesting). We use OFDMA scheme and hence assume that the entire spectrum is equally segmented into $K_b$ non-overlapping channels with each device occupying a given channel. Further, in presence of Eve, the transmitter at each DA port uses Wyner's wiretap coding and the secrecy rate achievable for the $k^{th}$ Bob \cite{wyner1975wire}: 
\begin{equation}
\label{eqn:1}
\begin{split}
R^k_s=\frac{1}{K_b}&\Bigg[\log_2\left(1+\Delta^b_kB_k\right)-\log_2\left(1+\Delta^eE\right)\Bigg]\\
%\text{where:}\hspace{0.2em}&E_l=\sum_{j=1}^{N}\Gamma^e_{j,l}
\end{split}
\end{equation}
In equation (1), $\Delta^b_k$ and $\Delta^e$ denote the PS ratios of $k^{th}$ Bob and the Eve respectively. $B_k=\sum_{j=1}^{N}\gamma^b_{j,k}p_{j,k}$ and $E=\sum_{j=1}^{N}\gamma^e_{j}p_{j,k}$, where $\gamma^b_{j,k}$ and $\gamma^e_{j}$ are the effective channel gain to noise power ratios from $j^{th}$ DA-port to $k^{th}$ Bob and the Eve respectively. $p_{j,k}$ is the transmit power from $j^{th}$ port to $k^{th}$ device.  \\
Now, secure energy efficiency is defined as:
\begin{equation} 
\label{eqn:2}
\eta_{SEE}=\frac{R_{total}}{P_{total}}=\frac{\sum_{k=1}^{K_b}R^k_s}{\sum_{k=1}^{K_b}\sum_{i=1}^{N}p_{i,k}+p_c}
\end{equation}
where $p_c$ is the power consumed in DA-ports during various signal processing operations.
Since each IoT device can decode the information from a given channel, but can harvest energy from all the available channels, the energy harvested by $k^{th}$ Bob can be expressed as:
%\begin{equation} 
$E^b_k=\tau^b_k(1-\Delta^b_k)\sum_{i=1}^{N}\gamma^b_{i,k}\sum_{j=1}^{K_b}p_{i,j}$
%\end{equation}
where, $\tau_k$ is the linear energy conversion efficiency of $k^{th}$ Bob.
Similarly, the energy harvested by the Eve is given by
%\begin{equation} 
$E^e=\tau^e(1-\Delta^e)\sum_{i=1}^{N}\gamma^e_{i}\sum_{j=1}^{K_b}p_{i,j}$
%\end{equation}
where, $\tau^e$ is the corresponding energy conversion efficiency of the Eve.
%\begin{table}[]
%\footnotesize
%\centering
%\caption{An example of a table.}
%\begin{tabular}{@{}lcc@{}}
%\toprule
%Column heading & Column A & Column B \\ \midrule
%And an entry                   & 1        & 2        \\
%And another entry              & 3        & 4        \\
%And another entry              & 5        & 6        \\
%And another entry              & 7        & 8        \\ \bottomrule
%\end{tabular}
%\label{tab:ex}
%\end{table}

\subsection{Problem formulation}
Our objective is to optimally allocate power to the DA-ports in order to maximize $\eta_{SEE}$ in (2) while meeting certain constraints.
To this end, we formulate the maximization of $\eta_{SEE}$ as a constrained fractional optimization problem as following:
\begin{equation}
\begin{split}
P1:&\max_{\{P_{i,k}\}} \eta_{SEE}\\
\text{s.t:}\hspace{0.5em}&\text{C1:}\sum_{k=1}^{K_b}p_{i,k} \leq P_{max,i},\hspace{0.3em} \text{C2:}\hspace{0.3em} p_{i,k}\geq 0\\
&\text{C3:}\hspace{0.3em}E^b_k\geq E^{b}_{k,min} \hspace{0.3em} \text{and} \hspace{0.3em}\text{C4:}\hspace{0.3em}E^e\leq E^{e}_{min}\\&\text{for}\hspace{0.4em}k=1,..,K_b ,\hspace{0.3em}i=1,..,N
\end{split}
\end{equation}
where, C1 and C2 correspond to the maximum and minimum transmit power constraints respectively, C3 is the constraint of minimum harvested energy ($E^b_{k,min}$) for $k^{th}$ Bob and the novel constraint C4 limits the energy harvested by the Eve. We introduce the constraint C4 in the problem in order to restrain the Eve from harvesting the energy from the received signal, thereby, restricting it's battery charge. 
Now, as is customary in PHY security literature, we assume that transmitter-Bob channel is better than transmitter-Eve channel (such that, $\Delta^b_kB_k> \Delta^eE$). Therefore, $R^k_s$ in (1) is a concave function. Thus, it is easy to verify that P1 is a concave linear fractional problem with pseudo-concave objective function \cite{zappone2015energy}. Hence, each stationary point is the global maximizer and KKT conditions are necessary and sufficient for optimality. For detailed proof refer to \cite{zappone2015energy} .
Since, the objective function in (3) is twice differentiable, we use Sequential Quadratic Programming (SQP) to solve the KKT conditions for the optimal solution \cite{nocedal2006sequential}. The numerical results are discussed in section 4.
\section{Outage probability of SEE}
Now, let's consider the case, when the CSI of Bob and Eve is not available at the transmitter. We characterize the system performance by the outage probability (OP) of SEE. We consider the blanket transmission scheme, with all DA-ports transmitting at same power level ($p$). Let $h_i$ and $g_i$ denote the independent and identically distributed (IID) circularly symmetric complex Gaussian (CSCG) channel coefficients (of Bob and Eve respectively) with zero mean and unit variance. Also, let $\sigma_b^2$ and $\sigma_e^2$ denote the noise variances at Bob and Eve respectively.  Therefore, instantaneous SNRs at Bob and Eve are given by $X'=\Delta^b\sum_{i=1}^{N}|h_i|^2p_i/\sigma_b^2=\Delta^b(p/\sigma_b^2)\sum_{i=1}^{N}|h_i|^2$ and $Y'=\Delta^e\sum_{i=1}^{N}|g_i|^2p_i/\sigma_e^2=\Delta^e(p/\sigma_e^2)\sum_{i=1}^{N}|g_i|^2$ respectively. Let $w^b=\Delta^b(p/\sigma_b^2)$, $X''=\sum_{i=1}^{N}|h_i|^2$, $w^e=\Delta^e(p/\sigma_e^2)$ and $Y''=\sum_{i=1}^{N}|g_i|^2$. \\The OP of SEE corresponding to a given threshold ($\eta_{th}$) is therefore given by:  $P_{out}(\eta_{th})=P(\eta_{SEE} < \eta_{th})$
\begin{equation} 
\begin{split}
= P&\left[\frac{\log_2\left(1+w^bX''\right)-\log_2\left(1+w^eY''\right)}{Np+p_c} < \eta_{th}\right] \\
= P&\left(w^bX'' < \{1+w^eY''\}\{2^{(Np+P_c)\eta_{th}}\}-1\triangleq Q\right)
\end{split}
\end{equation} 
\vspace{-1em}
\begin{theorem}
	For the given SWIPT-in-DAS setup $P_{out}(\eta_{th})$ is given by equation (6).
\end{theorem}
%\vspace{-2em}
\begin{proof}
	\vspace{-0.51em}
	Since $h_i$ and $g_i$ are IID-CSCG random variables, $|h_i|^2$ and $|g_i|^2$ will be exponentially distrbuted and hence $X''$ and $Y''$, being the sum of $N$ independent exponential random variables, will both follow Erlang distribution \cite{akkouchi2008convolution}. Let $X=w^bX''\sim f_X(x)= \frac{x^{N-1}e^{\frac{x}{w^b}}}{(w^b)^N(N-1)!}=\frac{\alpha^Nx^{N-1}e^{-\alpha x}}{(N-1)!}$ and $Y=w^eY''\sim f_Y(y)= \frac{y^Ne^{\frac{y}{w^e}}}{(w^e)^N(N-1)!}=\frac{\beta^Ny^{N-1}e^{-\beta y}}{(N-1)!}$, where $\alpha=\frac{1}{w^b}$ and $\beta=\frac{1}{w^e}$. Therefore we have:\\
	\vspace{-0.8em}
	\begin{equation}\nonumber
	\begin{split}
	&P(X\leq Q)=\int_{0}^{\infty}\int_{o}^{Q}f_X(x)f_Y(y)dxdy\\
	&=\int_{0}^{\infty}\int_{o}^{Q}\frac{\alpha^Nx^{N-1}e^{-\alpha x}}{(N-1)!}f_Y(y)dxdy\\
	&=\int_{0}^{\infty}\left(1-\sum_{n=0}^{N-1}\frac{(\alpha Q)^ne^{(-\alpha Q)}}{n!}\right)f_Y(y)dy\\
	\end{split}
	\end{equation} 
	\vspace{-1em}
	Let $z=\{Np+P_c\}\eta_{th}$, therefore;
	\vspace{0.3em}
	\begin{equation}\nonumber
	\begin{split}
	P_{out}&(\eta_{th})=1-\sum_{n=0}^{N-1}\int_{0}^{\infty}\left[\frac{(\alpha^n)(2^z+2^zy-1)^n}{\{e^{\alpha(2^z-1)}e^{\alpha2^zy}\}n!}\right]f_Y(y)dy\\
	%\end{split}
	%\end{equation}
	%\begin{equation}\nonumber
	&=1-K\sum_{n=0}^{N-1}\frac{\alpha^n}{n!}\left[\int_{0}^{\infty}(ay+b)^ne^{-(a\alpha+\beta)y}y^{N-1}dy\right]\\
	%\end{split}
	%\end{equation}
	&\text{where,} K=\frac{\beta^Ne^{\alpha(1-2^z)}}{(N-1)!},\hspace{0.4em}a=2^z,\hspace{0.4em}b=2^z-1\\
	%\begin{equation}\nonumber
	%\begin{split}
	%&\implies P_{out}(\eta_{th})=\\
	&=1-K\sum_{n=0}^{N-1}\frac{\alpha^n}{n!}\left[\int_{0}^{\infty}\sum_{j=0}^{n}\binom{n}{j}\frac{(ay)^jb^{n-j}}{e^{(a\alpha+\beta)y}}y^{N-1}dy\right]\\
	&=1-K\sum_{n=0}^{N-1}\frac{\alpha^n}{n!}\left\{\sum_{j=0}^{n}\binom{n}{j}b^{n-j}a^j\right\}\int_{0}^{\infty}\frac{y^{j+N-1}}{e^{(a\alpha+\beta)x}}dx\\
	&=1-K\sum_{n=0}^{N-1}\frac{\alpha^n}{n!}\left\{\sum_{j=0}^{n}\binom{n}{j}b^{n-j}a^j\right\}\left\{\frac{(j+N-1)!}{(a\alpha+\beta)^{j+N}}\right\}
	\end{split}
	\end{equation}
	Using $K$, $a$, $b$ and $z$ defined above, we get (4).
\end{proof}
%\begin{corollary} 
If only the Eve's CSI is unknown, we can have a special case for the OP of SEE as given below:
\begin{equation} \nonumber
\begin{split}
&P_{out}(\eta_{th})= P\left(Y >\frac{\left(1+w^bX''\right)}{2^{(Np+p_c)\eta_{th}}}-1\triangleq Q'\right)\\
&=1-\int_{0}^{Q'}\frac{\beta^Ny^{N-1}e^{-\beta y}}{(N-1)!}dy %=1-\frac{\int_{0}^{Q'} \left[\beta^Ny^{N-1}e^{-\beta y}\right]dy}{(N-1)!}\\
=1-\frac{\gamma_{inc}(N,\beta Q')}{(N-1)!}
\end{split}
\end{equation}
\begin{equation}
\begin{split}
\implies	P_{out}(\eta_{th})=1-\frac{\gamma_{inc}\{N,\frac{1}{w^e} \frac{(1+w^b\sum_{i=1}^{N}|h_i|^2)}{(2^{(Np+P_c)\eta_{th}})}-1\}}{(N-1)!}
\end{split}
\end{equation}
where, $\gamma_{inc}(N,x)$ represents the lower incomplete gamma function.
%\end{corollary}
%\begin{proof} $P_{out}(\eta_{th})=P(\eta_{SEE} < \eta_{th})$\\ 

\vspace{-1em}
%\end{proof}
%\vspace{-1em}
\subsection{Outage probability of worst case SEE}
If there are $M$ eavesdroppers in the system, the overall performance of the system is determined by the worst case secrecy rate achievable for the given user. 
\begin{corollary}
	The OP of worst case SEE is given by $P_{out}(\eta_{th})= 1-\prod_{m=1}^{M} P\left(Y_m < Q''\right)$
\end{corollary}
\begin{proof}
	Considering the secrecy rate corresponding to the maximum of $M$ eavesdroppers, we have: $P_{out}(\eta_{th})=P(\eta_{SEE} < \eta_{th})$
	\begin{equation} \nonumber
	\begin{split}
	=& P\left[\frac{\log_2\left(1+X\right)-\max_{m \in M}\log_2\left(1+Y_m\right)}{Np+p_c} < \eta_{th}\right] \\
	=& P\left[ \max_{m \in M}Y_m>\frac{\left(1+X\right)}{2^{(Np+p_c)\eta_{th}}}-1\triangleq Q''\right] \\
	=& 1-P\left(\max_{m \in M}Y_m < Q''\right)=1-\prod_{m=1}^{M} P\left(Y_m < Q''\right)
	\end{split}
	\end{equation}
	%\begin{equation}
	%$P_{out}(\eta_{th})= 1-\prod_{m=1}^{M} P\left(Y_m < Q''\right)$
	%\end{equation}
	where, $Y_m's$ are assumed independent. $P\left(Y_m < Q''\right)=\int_{0}^{\infty}\int_{0}^{Q''}f_{Y_m}(y)f_X(x)dydx$, can be evaluated in closed form similar to proof of Theorm 1 and is given by equation (7). 
\end{proof}
\begin{figure*}[h]
	\hrule
	\begin{equation} 
	%\begin{split}
	P_{out}(\eta_{th})=1-\frac{e^{\alpha\{1-2^{(N.p+Pc)\eta_{th}}\}}}{\beta^{-N}(N-1)!}\sum_{n=0}^{N-1}\frac{\alpha^n}{n!}\left[\sum_{j=0}^{n}\binom{n}{j}\left(2^{(N.p+Pc)\eta_{th}}-1\right)^{n-j}\frac{\left(2^{(N.p+Pc)\eta_{th}}\right)^j(j+N-1)!}{\{\alpha.2^{(N.p+Pc)\eta_{th}}+\beta\}^{j+N}}\right]
	%\end{split}
	\end{equation}
	\begin{equation}
	P\left(Y_m < Q''\right)=1-\frac{e^{\beta\{1-2^{-(N.p+Pc)\eta_{th}}\}}}{\alpha^{-N}(N-1)!}\sum_{n=0}^{N-1}\frac{\beta^n}{n!}\left[\sum_{j=0}^{n}\binom{n}{j}\left(2^{-(N.p+Pc)\eta_{th}}-1\right)^{n-j}\frac{\left(2^{-(N.p+Pc)\eta_{th}}\right)^j(j+N-1)!}{\{\beta.2^{-(N.p+Pc)\eta_{th}}+\alpha\}^{j+N}}\right]
	\end{equation}
	\hrule
\end{figure*}
\vspace{-1em}
\section{Results and discussion}
In this section, we present the numerical results of the optimization problem discussed in section 2 and the OP of SEE discussed in section 3. In Fig. 1 (a), we plot energy efficiency (EE) as a function of $P_{max}$ of DA-ports for different values of $E^b_{min}$ and with $\Delta^b_k=\Delta^e=0.5$, $\tau^b_k=\tau^e=0.75$ $\forall k$ and $N=6$. We observe that EE of the system initially improves with $P_{max}$ but eventually gets saturated. We also note that EE decreases as $E_{min}$ of the IoT device increases. Moreover, a comparison is provided between a system without security (WoS) and that with security (WS). In Fig. 1(b), we have the results corresponding to the charge constraint (CC) of EHE, with $E^b_{k,min}=1mW$  $\forall k$. We observe that, in SWIPT environment it is actually beneficial to have both Bob and Eve as energy harvesting nodes. However, we note that restraining the eavesdropper from charging is possible only at the cost of energy efficiency of the system. In Fig. 2, we plot OP of SEE w.r.t transmit power of DA-ports, with $\sigma^2_b$=-20dBm and $\sigma^2_e$=-10dBm. Transmit power and number of DA-ports play an important role in the overall system performance. In Fig. 2(a), we observe that $P_{out}(\eta_{th})$ reduces significantly as the transmit power is increased and initially, a similar trend is observed with the number of DA-ports. However, this behaviour changes at higher power levels. In fact, the results reveal that in order to minimize the OP of SEE, signals need to be transmitted at lower power levels when there are larger number of DA-ports. Further, in Fig. 2 (b), it can be observed that $P_{out}(\eta_{th})$ also decreases when PS ratios of Bob and Eve are increased. Moreover, the system performs better when the PS ratio of Bob is higher in magnitude than that of the Eve.   

\begin{figure}[h]
	\centering
	\subfigure[]{\includegraphics[width=.493\columnwidth]{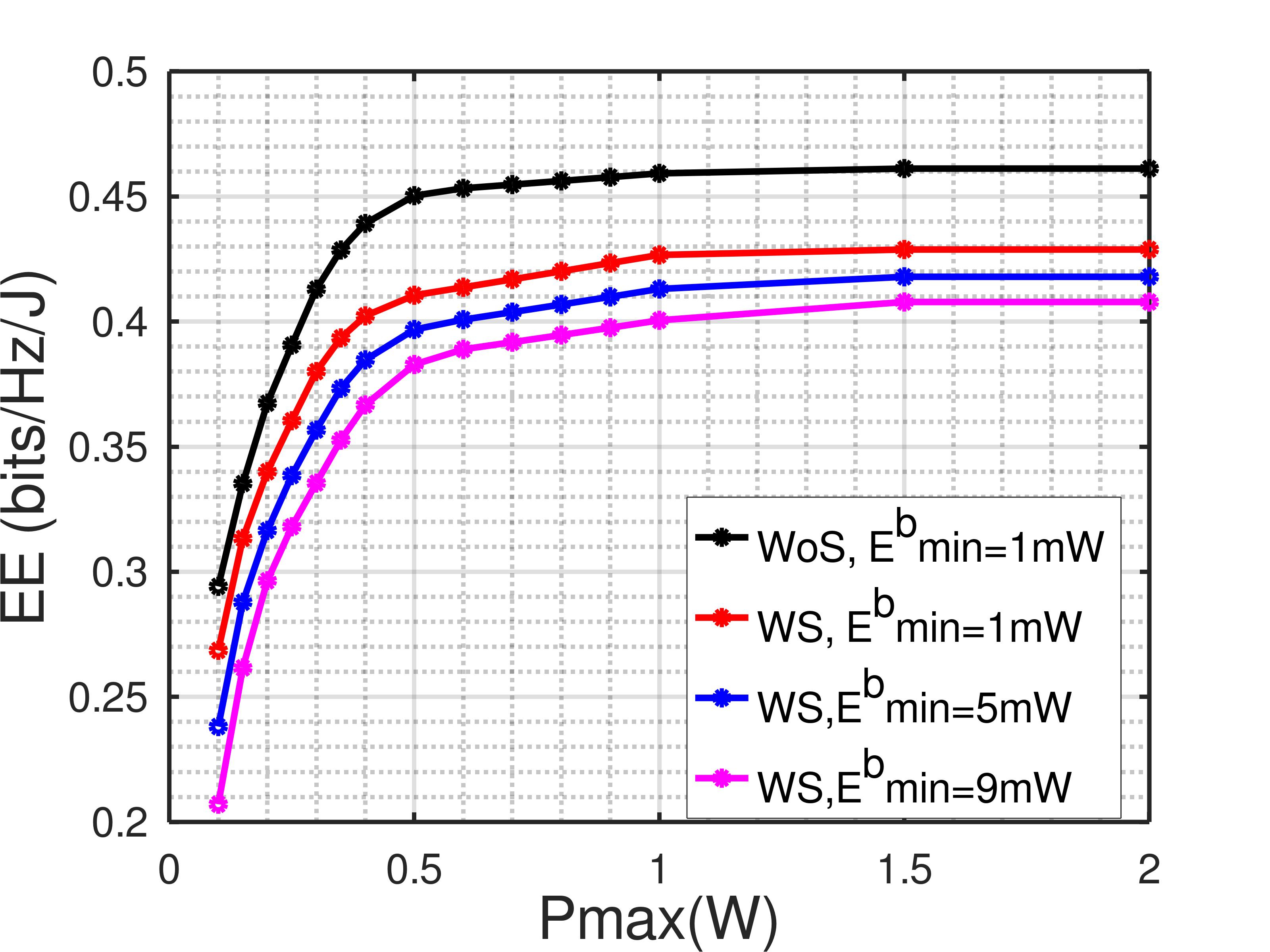}}
	\subfigure[]{\includegraphics[width=.493\columnwidth]{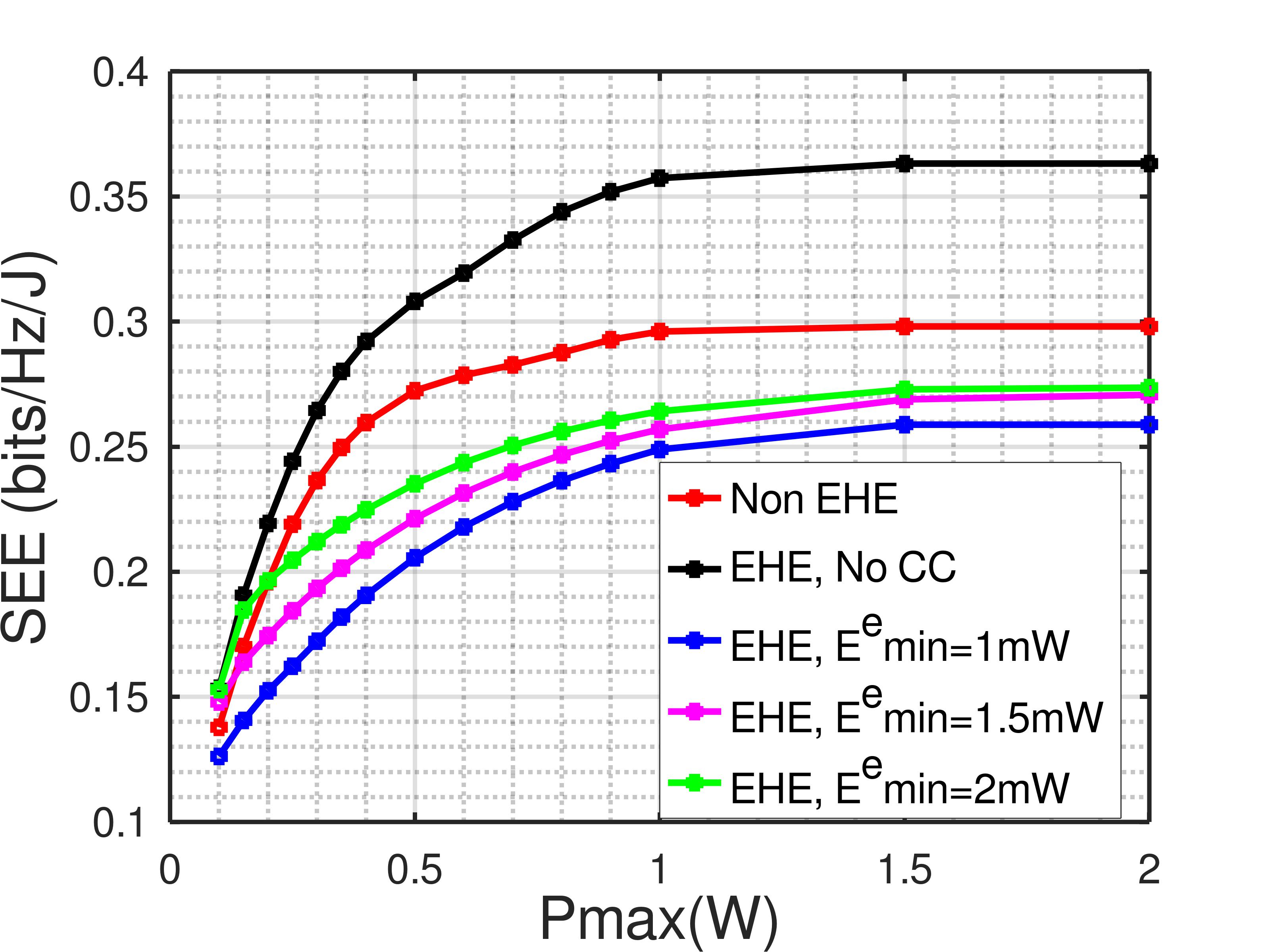}}
	\caption{(a) Enegy Efficiency (EE) w.r.t $P_{max}$ for different values of $E^b_{min}$ (b) SEE w.r.t $P_{max}$ in case of an EHE}
	%\label{fig:ex}
\end{figure}
%\vspace{-2em}
\begin{figure}[h]
	\centering
	\subfigure[]{\includegraphics[width=.493\columnwidth]{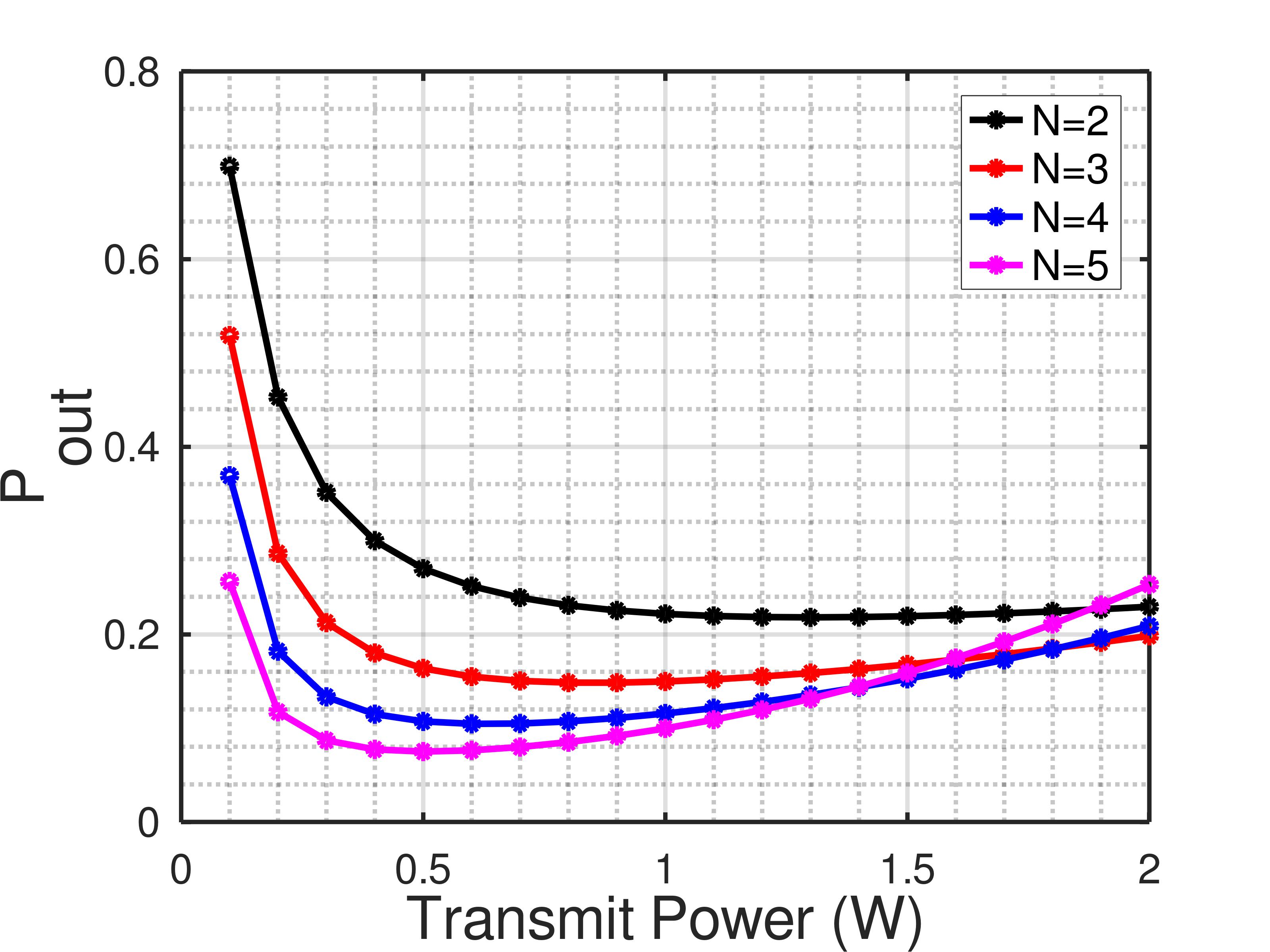}}
	\subfigure[]{\includegraphics[width=.493\columnwidth]{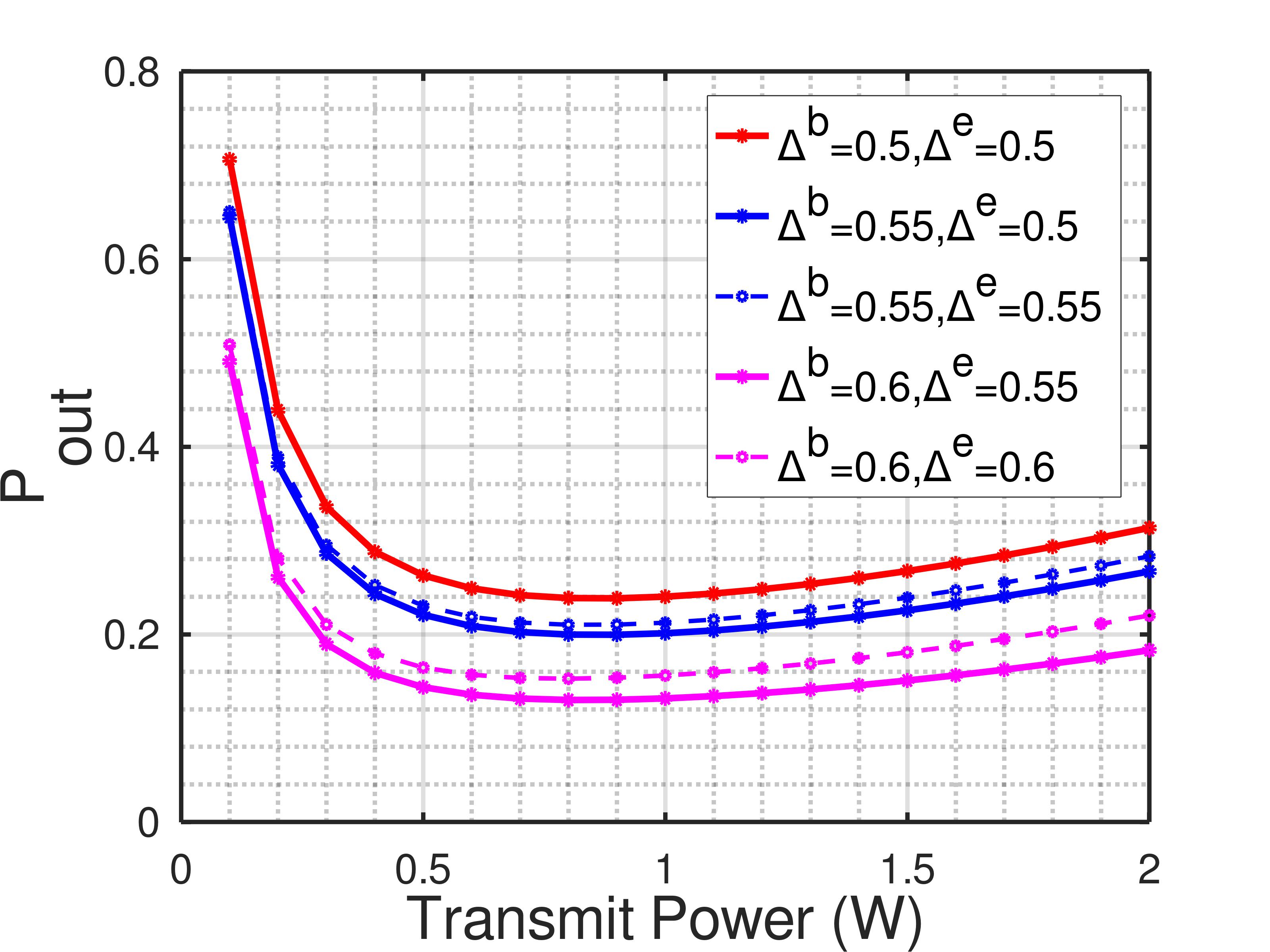}}
	\caption{(a) OP of SEE w.r.t transmit power for increasing $N$ (b) OP of SEE w.r.t transmit power for increasing $\Delta^b$ and $\Delta^e$} 
\end{figure}
\vspace{-1em}
\section{Conclusion}
In this paper, we studied SEE for SWIPT-in-DAS based IoT network. With the assumption of perfect CSI, our objective was to study the effect of an EHE over the maximization of SEE. 
Further, in an unknown CSI scenario, we characterized the system performance in terms of outage probability of SEE. For the blanket transmission scheme, we obtained a closed form expression for the OP of SEE. The theoretical results obtained were supported with the numerical computations.  


\bibliographystyle{IEEEtran}
\bibliography{references}

\end{document}